\documentclass{article}
\usepackage{helvet}
\usepackage{courier}
\usepackage[T1]{fontenc}
\usepackage[a4paper]{geometry}
\usepackage[active]{srcltx}
\usepackage{color}
\usepackage{verbatim}
\usepackage{cprotect}
\usepackage{mathrsfs}
\usepackage{mathtools}
\usepackage{dsfont}
\usepackage{amsmath}
\usepackage{amsthm}
\usepackage{amssymb}
\usepackage{makeidx}
\makeindex
\usepackage[pdfusetitle,
 bookmarks=true,bookmarksnumbered=true,bookmarksopen=true,bookmarksopenlevel=2,
 breaklinks=false,pdfborder={0 0 1},backref=page,colorlinks=false]
 {hyperref}
\hypersetup{
 unicode=false, linkcolor=black, citecolor=black, urlcolor=blue, filecolor=blue, pdfpagelayout=OneColumn, pdfnewwindow=true, pdfstartview=XYZ, plainpages=false}

\makeatletter
\numberwithin{equation}{section}
\theoremstyle{plain}
\newtheorem{thm}{\protect\theoremname}[section]
\theoremstyle{plain}
\newtheorem{question}[thm]{\protect\questionname}
\theoremstyle{plain}
\newtheorem{lem}[thm]{\protect\lemmaname}
\theoremstyle{plain}
\newtheorem{lyxalgorithm}[thm]{\protect\algorithmname}
\theoremstyle{definition}
\newtheorem{example}[thm]{\protect\examplename}
\theoremstyle{definition}
\newtheorem{defn}[thm]{\protect\definitionname}
\theoremstyle{remark}
\newtheorem{rem}[thm]{\protect\remarkname}

\usepackage{color}
\usepackage{slashed}
 \let\myTOC\tableofcontents
 \renewcommand\tableofcontents{%
   \pdfbookmark[1]{Contents}{}
   \myTOC
   \cleardoublepage
   \pagenumbering{arabic} }

\global\long\def\foreignlanguage#1#2{#2}%
\global\long\def\selectlanguage#1{}%

\allowdisplaybreaks

\DeclareMathOperator\supp{\mathrm{supp}}
\DeclareMathOperator\Tr{\mathrm{Tr}}
\DeclareMathOperator\tr{\mathrm{Tr}}

\makeatother

\providecommand{\algorithmname}{Algorithm}
\providecommand{\definitionname}{Definition}
\providecommand{\examplename}{Example}
\providecommand{\lemmaname}{Lemma}
\providecommand{\questionname}{Question}
\providecommand{\remarkname}{Remark}
\providecommand{\theoremname}{Theorem}

\begin{document}
\title{What is electric charge?: Charge as a local observable in relativistic
quantum field theory}
\author{Hideyasu Yamashita\\
{\small Division of Liberal Arts and Sciences, Aichi-Gakuin University}\\
\texttt{\small yamasita@dpc.aichi-gakuin.ac.jp}
}

\date{\today}

\maketitle
\newcommand{\dcolor}{\definecolor{note_fontcolor}{rgb}{0.1, 0.0, 0.8}}
\definecolor{HYnote_fontcolor}{rgb}{0.2, 0.0, 0.2}
\definecolor{hycolor}{rgb}{0.3, 0.0, 0.3}
\newcommand{\hyc}{\color{hycolor}}
\newenvironment{HYnote}
 {\textcolor{note_fontcolor}\bgroup\ignorespaces}
  {\ignorespacesafterend\egroup} 

\newenvironment{trivenv}
  {\bgroup\ignorespaces}
  {\ignorespacesafterend\egroup}

\newcommand{\displabel}[1]{}

\newcommand{\hidable}[3]{#2}
\newcommand{\hidea}[1]{{#1}}
\newcommand{\hideb}[1]{{#1}}
\newcommand{\hidec}[1]{{#1}}
\newcommand{\hidep}[1]{{#1}}
\renewcommand{\hidec}[1]{}
\renewcommand{\hidep}[1]{}

\newcommand{\thlab}[1]{{\tt [#1]}}

\newcommand{\black}{\color{black}}

\global\long\def\N{\mathbb{N}}%
\global\long\def\C{\mathbb{C}}%
\global\long\def\Z{\mathbb{Z}}%
 
\global\long\def\R{\mathbb{R}}%
 
\global\long\def\im{\mathrm{i}}%

\global\long\def\di{\partial}%
 
\global\long\def\d{{\rm d}}%

\global\long\def\ol#1{\overline{#1}}%
\global\long\def\ul#1{\underline{#1}}%
\global\long\def\ob#1{\overbrace{#1}}%

\global\long\def\ov#1{\overline{#1}}%

\global\long\def\then{\Rightarrow}%
 
\global\long\def\Then{\Longrightarrow}%

\global\long\def\N{\mathbb{N}}%
\global\long\def\C{\mathbb{C}}%
\global\long\def\Z{\mathbb{Z}}%
 
\global\long\def\R{\mathbb{R}}%
 
\global\long\def\im{\mathrm{i}}%

\global\long\def\di{\partial}%
 
\global\long\def\d{{\rm d}}%

\global\long\def\ol#1{\overline{#1}}%
\global\long\def\ul#1{\underline{#1}}%
\global\long\def\ob#1{\overbrace{#1}}%

\global\long\def\ov#1{\overline{#1}}%

\global\long\def\then{\Rightarrow}%
 
\global\long\def\Then{\Longrightarrow}%

\global\long\def\cA{\mathcal{A}}%
\global\long\def\cB{\mathcal{B}}%
 
\global\long\def\cC{\mathcal{C}}%
 
\global\long\def\cD{\mathcal{D}}%
\global\long\def\cE{\mathcal{E}}%
 
\global\long\def\cF{\mathcal{F}}%
 
\global\long\def\cG{{\cal G}}%
 
\global\long\def\cH{\mathcal{H}}%
 
\global\long\def\cI{\mathcal{I}}%
 
\global\long\def\cJ{\mathcal{J}}%
\global\long\def\cK{\mathcal{K}}%
 
\global\long\def\cL{\mathcal{L}}%
 
\global\long\def\cM{\mathcal{M}}%
 
\global\long\def\cN{\mathcal{N}}%
 
\global\long\def\cO{\mathcal{O}}%
 
\global\long\def\cP{\mathcal{P}}%
 
\global\long\def\cQ{\mathcal{Q}}%
 
\global\long\def\cR{\mathcal{R}}%
 
\global\long\def\cS{\mathcal{S}}%
 
\global\long\def\cT{\mathcal{T}}%
 
\global\long\def\cU{\mathcal{U}}%
 
\global\long\def\cV{\mathcal{V}}%
 
\global\long\def\cW{\mathcal{W}}%
\global\long\def\cX{\mathcal{X}}%
 
\global\long\def\cY{\mathcal{Y}}%
 
\global\long\def\cZ{\mathcal{Z}}%

\global\long\def\scA{\mathscr{A}}%
\global\long\def\scB{\mathscr{B}}%
\global\long\def\scC{\mathscr{C}}%
\global\long\def\scD{\mathscr{D}}%
 
\global\long\def\scE{\mathscr{E}}%
 
\global\long\def\scF{\mathscr{F}}%
 
\global\long\def\scG{\mathscr{G}}%
 
\global\long\def\scH{\mathscr{H}}%
 
\global\long\def\scI{\mathscr{I}}%
 
\global\long\def\scJ{\mathscr{J}}%
 
\global\long\def\scK{\mathscr{K}}%
 
\global\long\def\scL{\mathscr{L}}%
 
\global\long\def\scM{\mathscr{M}}%
 
\global\long\def\scN{\mathscr{N}}%
 
\global\long\def\scO{\mathscr{O}}%
 
\global\long\def\scP{\mathscr{P}}%
 
\global\long\def\scR{\mathscr{R}}%
\global\long\def\scS{\mathscr{S}}%
 
\global\long\def\scT{\mathscr{T}}%
 
\global\long\def\scU{\mathscr{U}}%
 
\global\long\def\scW{\mathscr{W}}%
\global\long\def\scZ{\mathscr{Z}}%

\global\long\def\bbA{\mathbb{A}}%
 
\global\long\def\bbB{\mathbb{B}}%
 
\global\long\def\bbD{\mathbb{D}}%
 
\global\long\def\bbE{\mathbb{E}}%
 
\global\long\def\bbF{\mathbb{F}}%
 
\global\long\def\bbG{\mathbb{G}}%
 
\global\long\def\bbI{\mathbb{I}}%
 
\global\long\def\bbJ{\mathbb{J}}%
 
\global\long\def\bbK{\mathbb{K}}%
 
\global\long\def\bbL{\mathbb{L}}%
 
\global\long\def\bbM{\mathbb{M}}%
 
\global\long\def\bbP{\mathbb{P}}%
 
\global\long\def\bbQ{\mathbb{Q}}%
 
\global\long\def\bbT{\mathbb{T}}%
 
\global\long\def\bbU{\mathbb{U}}%
 
\global\long\def\bbX{\mathbb{X}}%
 
\global\long\def\bbY{\mathbb{Y}}%
\global\long\def\bbW{\mathbb{W}}%

\global\long\def\bbOne{1\kern-0.7ex  1}%
 %defined as 1\kern-0.7ex1

\renewcommand{\bbOne}{\mathbbm{1}}

\global\long\def\bB{\mathbf{B}}%
 
\global\long\def\bG{\mathbf{G}}%
 
\global\long\def\bH{\mathbf{H}}%
\global\long\def\bS{\boldsymbol{S}}%
 
\global\long\def\bT{\mathbf{T}}%
 
\global\long\def\bX{\mathbf{X}}%
\global\long\def\bY{\mathbf{Y}}%
\global\long\def\bW{\mathbf{W}}%
 
\global\long\def\boT{\boldsymbol{T}}%

\global\long\def\fraka{\mathfrak{a}}%
 
\global\long\def\frakb{\mathfrak{b}}%
 
\global\long\def\frakc{\mathfrak{c}}%
 
\global\long\def\frake{\mathfrak{e}}%
 
\global\long\def\frakf{\mathfrak{f}}%
 
\global\long\def\fg{\mathfrak{g}}%
 
\global\long\def\frakh{\mathfrak{h}}%
 
\global\long\def\fraki{\mathfrak{i}}%
\global\long\def\frakk{\mathfrak{k}}%
 
\global\long\def\frakl{\mathfrak{l}}%
 
\global\long\def\frakm{\mathfrak{m}}%
 
\global\long\def\frakn{\mathfrak{n}}%
 
\global\long\def\frako{\mathfrak{o}}%
 
\global\long\def\frakp{\mathfrak{p}}%
 
\global\long\def\frakq{\mathfrak{q}}%
 
\global\long\def\fraks{\mathfrak{s}}%
 
\global\long\def\fs{\mathfrak{s}}%
 
\global\long\def\fraku{\mathfrak{u}}%
\global\long\def\frakz{\mathfrak{z}}%

\global\long\def\fA{\mathfrak{A}}%
 
\global\long\def\fB{\mathfrak{B}}%
 
\global\long\def\fC{\mathfrak{C}}%
 
\global\long\def\fD{\mathfrak{D}}%
 
\global\long\def\fF{\mathfrak{F}}%
 
\global\long\def\fG{\mathfrak{G}}%
 
\global\long\def\fK{\mathfrak{K}}%
 
\global\long\def\fL{\mathfrak{L}}%
 
\global\long\def\fM{\mathfrak{M}}%
 
\global\long\def\fP{\mathfrak{P}}%
 
\global\long\def\fR{\mathfrak{R}}%
 
\global\long\def\fS{\mathfrak{S}}%
\global\long\def\fT{\mathfrak{T}}%
 
\global\long\def\fU{\mathfrak{U}}%
\global\long\def\fV{\mathfrak{V}}%
 
\global\long\def\fW{\mathfrak{W}}%
\global\long\def\fX{\mathfrak{X}}%
\global\long\def\fZ{\mathfrak{Z}}%

\global\long\def\ssS{\mathsf{S}}%
\global\long\def\ssT{\mathsf{T}}%
 
\global\long\def\ssW{\mathsf{W}}%

\global\long\def\rM{\mathrm{M}}%
\global\long\def\prj{\mathfrak{P}}%

{} 
\global\long\def\sy#1{{\color{blue}#1}}%

\global\long\def\magenta#1{{\color{magenta}#1}}%

% \global\long\def\symb#1{{\color{red}#1}}%
\global\long\def\symb#1{#1}%

\global\long\def\emhrb#1{\text{{\color{red}{\huge {\bf #1}}}}}%

\newcommand{\symbi}[1]{\index{$ #1$}{\color{red}#1}} 

{} 
% \global\long\def\SYM#1#2{\symb{#1}_{\##2}}%
\global\long\def\SYM#1#2{#1}%

\renewcommand{\SYM}[2]{\symb{#1}}

\newcommand{\usuji}{\color[rgb]{0.7,0.4,0.4}} \newcommand{\usu}{\color[rgb]{0.5,0.2,0.1}}
\newenvironment{Usuji} {\begin{trivlist}   \item \usuji }  {\end{trivlist}}
\newenvironment{Usu} {\begin{trivlist}   \item \usu }  {\end{trivlist}} 

\newcommand{\term}[1]{\textcolor[rgb]{0, 0, 1}{\bf #1}}
\newcommand{\termi}[1]{{\bf #1}}

\newcommand{\slim}{\mathop{\mbox{s-lim}}} %

\newcommand{\wlim}{\mathop{\mbox{w-lim}}}

\newcommand{\limsub}{\mathop{\mbox{\rm lim-sub}}}

\global\long\def\bboxplus{\boxplus}%

\renewcommand{\bboxplus}{\mathop{\raisebox{-0.8ex}{\text{\begin{trivenv}\LARGE{}$\boxplus$\end{trivenv}}}}}

\global\long\def\upha{\upharpoonright}%

\global\long\def\ket#1{|#1\rangle}%
 
\global\long\def\bra#1{\langle#1|}%

{} 
\global\long\def\lll{\vert\kern-0.25ex  \vert\kern-0.25ex  \vert}%
 \renewcommand{\lll}{{\vert\kern-0.25ex  \vert\kern-0.25ex  \vert}}

\global\long\def\biglll{\big\vert\kern-0.25ex  \big\vert\kern-0.25ex  \big\vert\kern-0.25ex  }%
 
\global\long\def\Biglll{\Big\vert\kern-0.25ex  \Big\vert\kern-0.25ex  \Big\vert}%

\newcommand{\iiia}[1]{{\left\vert\kern-0.25ex\left\vert\kern-0.25ex\left\vert #1
  \right\vert\kern-0.25ex\right\vert\kern-0.25ex\right\vert}}

\global\long\def\iii#1{\iiia{#1}}%

% misc %%%%%%%%%%%%%%%%%%%%%%%%%%%%%%%%%%%%%%%%%%%%%%%%%%%%%

\global\long\def\Ae{{\rm a.e.}}%
\global\long\def\Ad{{\rm Ad}}%
\global\long\def\ad{{\rm ad}}%
\global\long\def\Borel{{\rm Borel}}%
\global\long\def\area{{\bf S}}%

\global\long\def\bbS{\mathbb{S}}%
\global\long\def\bOne{{\bf 1}}%
\global\long\def\bbOne{\mathds{1}}%
\global\long\def\Bdd{\mathscr{B}}%
\global\long\def\Borel{{\rm Borel}}%
\global\long\def\bP{{\bf P}}%

\global\long\def\Cov{\mathrm{Cov}}%

\global\long\def\Cl{{\rm C}\ell}%
\global\long\def\cconj{\blacklozenge}%
\global\long\def\cpt{{\rm c}}%
\global\long\def\cc{{\bf c}}%
\global\long\def\curve{\mathsf{C}}%

\global\long\def\CCRW{\mathcal{CCR}}%
\global\long\def\ccr{{\rm Ccr}}%
\global\long\def\CCR{{\rm CCR}}%
\global\long\def\CAR{{\rm CAR}}%
\global\long\def\ComplexS{J}%

\global\long\def\diam{\text{{\rm diam}}}%
\global\long\def\dom{\mathrm{dom}}%
\global\long\def\End{{\rm End}}%
\global\long\def\ex{{\rm ex}}%

\global\long\def\bE{\mathbf{E}}%
\global\long\def\Ex{\mathbb{E}}%

\global\long\def\GL{{\rm GL}}%
\global\long\def\grad{\mathrm{grad}}%
 
\global\long\def\Hom{\mathrm{Hom}}%

\global\long\def\DiracOpm{{\bf \mathsf{D}}}%
\global\long\def\Dconj{\mathfrak{d}}%
\global\long\def\DSpinors{\mathscr{C}}%
\global\long\def\even{{\rm even}}%

\global\long\def\GreenOp{\mathsf{S}}%
\global\long\def\hol{{\rm hol}}%

\global\long\def\Id{{\rm Id}}%
 
\global\long\def\id{{\rm id}}%

\newcommand{\Kahler}{K{\"a}hler}

\global\long\def\LBundle{\cL}%
 
\global\long\def\Leb{\text{{\rm Leb}}}%
 
\global\long\def\Lie{{\rm Lie}}%
 
\global\long\def\leng{\text{{\rm leng}}}%
 
\global\long\def\meas{\text{{\rm meas}}}%

\global\long\def\Manifold{\mathscr{X}}%
 
\global\long\def\Mat{{\rm Mat}}%

\global\long\def\nablaslash{\slashed{\nabla}}%

\global\long\def\ON{{\rm ON}}%
{} 
\global\long\def\OCpl{{\rm OC}}%
\global\long\def\Odd{{\rm Odd}}%

\global\long\def\Proj{{\rm Proj}}%
 
\global\long\def\Prob{\mathbb{P}}%
\global\long\def\Var{\mathrm{Var}}%
\global\long\def\Pow{\mathsf{P}}%
{}

\global\long\def\p{\mathbf{p}}%
 
\global\long\def\q{\mathbf{q}}%

\global\long\def\qdev{{\rm qd}}%
\global\long\def\Paths{\scE}%
 
\global\long\def\Pin{{\rm Pin}}%
 
\global\long\def\POVM{\mathsf{M}}%
\global\long\def\STime{\mathcal{M}}%

\global\long\def\quantity{\mathsf{q}}%
\global\long\def\Quantities{\mathscr{Q}}%
\global\long\def\ran{\mathop{\mathrm{ran}}}%

\global\long\def\rD{{\rm D}}%
\global\long\def\re{{\bf r}}%
\global\long\def\RDSPinors{\scR}%

\global\long\def\spec{{\rm spec}}%
 
\global\long\def\Sp{{\rm Sp}}%

\global\long\def\subset{\subseteq}%
\global\long\def\sgn{{\rm sgn}}%
\global\long\def\Span{{\rm span}}%
\global\long\def\Spin{{\rm Spin}}%
\global\long\def\spin{\mathfrak{spin}}%
\global\long\def\Sol{\mathsf{Sol}_{{\rm sc}}}%

\global\long\def\Ten{\bullet}%
{} %

\global\long\def\TT{\intercal}%
 \renewcommand{\TT}{\mathsf{T}}

\global\long\def\vac{{\rm vac}}%

{} %

\global\long\def\WICK#1{{\rm w}\{#1\}}%
 \renewcommand{\WICK}[1]{{:}#1{:}}

\global\long\def\weyl{\mathsf{W}}%
\global\long\def\x{\mathbf{x}}%
 
\global\long\def\y{\mathbf{y}}%

\def\foreignlanguage#1#2{#2}

%\clearpage

\let\ruleorig=\rule
\renewcommand{\rule}{\noindent\ruleorig}

\global\long\def\labelenumi{(\arabic{enumi})}%

\newcommand{\reff}{\ref}

\global\long\def\Borel{{\rm Borel}}%

\global\long\def\Q{\mathsf{Q}}%
\global\long\def\FR{{\rm FR}}%
\global\long\def\FRP{{\rm FRP}}%
\global\long\def\Im{\mathop{{\rm Im}}}%
\global\long\def\wcolon{{:}}%
\global\long\def\wick#1{\,{:}#1{:}\,}%

\global\long\def\Scope{{\rm Scope}}%
\global\long\def\cfunc{\zeta}%
\global\long\def\balpha{\boldsymbol{\alpha}}%
\global\long\def\Spec{{\rm Spec}}%
\global\long\def\E{\cE}%

\begin{abstract}
It is rather surprising that modern quantum physics does not appear
to have provided any clear answer to the simple question ``what is
electric charge?''. Even when the total charge operator $Q$ is well-defined,
the non-locality of $Q$ implies that it is not an observable in the
usual sense, which can be measured by a (local) experimental apparatus.
A candidate for the ``local version'' of charge operator is the
4-current operator $j=(j^{\mu})_{\mu=0,1,2,3}$. However, it is known
that the rigorous definition of $j$ is difficult in $(3+1)$-dimensional
Minkowski space. Although it was found that the current can be defined
in $(1+1)$-dimensions (Carey et al.), I argue that even when $j$
can be suitably defined, the interpretability of $j$ as the ``local
charge operator'' is dubious. Instead I return to Araki and Wyss
(1964), and propose the concept of ``scope-local charge'' $Q_{\cfunc}(P)$
for a ``scope'' $P$, expressed by a finite-dimensional projection.
I work in an abstract $C^{*}$-algebraic setting.
\end{abstract}

\section{Introduction}

It is rather surprising that modern quantum physics does not appear
to have provided any clear answer to the simple question ``what is
electric charge?''.\cprotect\footnote{Actually, the same applies to most of the fundamental concepts such
as ``position'', ``momentum'', ``time'', ``energy'', etc.
Further, all these concepts become less clear in curved spacetimes.
However, the concept of electric charge appears to be one of the most
difficult aspects to deal with in axiomatic and/or algebraic QFT.
For example, although the deep results of the DHR analysis (see \cite[Ch.IV]{Haa96}
and referenced therein) considerably clarified the concept of ``charge''
in some special cases (e.g., baryon number and lepton number), it
excludes states with electric charge from consideration.%
} It seems that there has been no firm, unified view on this issue
until now, even for the (free) Dirac field in Minkowski spacetime
(for recent arguments, see e.g., \cite{Tum2022,RT2024}). The notion
of charge in a free field will be more complicated in curved spacetime,
see e.g., \cite{DL2012,BBS2020}. Of course, the problem would become
far more difficult for interacting fields.

Here by a ``clear answer'', I mean an answer which has both the
following two properties:
\begin{enumerate}
\item Rigor: It is desirable that the (electric) charge can be defined as
a well-defined self-adjoint operator on a Hilbert space $\cH$; equivalently,
as a $\cP$-valued measure on $\Borel(\R)$ (Borel subsets of $\R$),
where $\cP$ denotes the set of projections on $\cH$. More generally,
we can also consider the cases where $\cP$ denotes the set of the
projections in a $C^{*}$-algebra or $W^{*}$-algebra, etc.
\item Empirically clarity: It is desirable that we can specify the operational/experimental
procedures to measure the charge in some ``scope''.
\end{enumerate}
Here, I used the term ``scope'' in the sense illustrated below.
A typical example is the ``photon number'' observable, which is
frequently used in quantum optics. Generally, ``the number of all
the photons in a bounded spacetime region'' is probably meaningless,
because of the existence of the ``infrared particles''; Empirically,
it is extremely difficult to count low-energy photons (e.g., in a
long wave) by a counter such as a photomultiplier tube; Mathematically,
we have also the difficulty concerning infrared divergence. Instead,
we usually consider the number of photons in a ``single mode'',
e.g., in a single laser beam, in quantum optics. Thus we do not consider
the number of \emph{all} photons in a spacetime region $\cO$, but
that of only some special sort of photons in $\cO$, in other words,
the photons found in a specific restricted ``scope''. Similarly,
perhaps the quantity ``the total charge in region $\cO$'' is not
definable both mathematically and empirically, and so instead I would
like to consider ``the total charge in a restricted scope''.

Although the condition (1) might not be a necessary condition of (2),
it is hard to conceive the situation where there is an answer which
does not satisfy (1), but do (2), because usually a precise description
of a measurement procedure needs a consistent mathematical language.

Of course, the condition (1) is not a \emph{sufficient} condition
of (2). For example, consider the (free) quantum Dirac field (of electrons/positrons),
formulated on a fermion Fock space $\cF=\cF(\cH_{+}\oplus\cH_{-})$,
where $\cH_{+}$ (resp.~$\cH_{-}$) is the single-positron (resp.~single-electron)
Hilbert space. Then the total charge operator $Q$ is easily defined
to be $Q:=N_{+}-N_{-}$, where $N_{+}$ and $N_{-}$ are the number
operators of positrons and electrons, respectively. (Precisely, instead
the operator $eQ$ ($e>0$: elementary charge) should be called the
charge operator.) However, it seems that the total charge $Q$ itself
has very little empirical significance; We can interpret $Q$ as the
``total charge of the whole Universe'', which cannot be measured
by any realistic measurement apparatus, in principle. Thus, although
the total charge operator $Q$ can be seen as a ``clear answer''
in the sense of (1), it is no good answer in the sense of (2). Actually,
we can only measure the \emph{local} observables. 

On the other hand, we know that it can be verified by a (local) experimental
procedure that any value of electric charge is an integer multiple
of the elementary charge $e$ (by e.g., Millikan's oil drop experiment.)
Let us call it the \termi{integer charge law}. Although the mathematical
fact that the spectrum of the total charge operator $Q$ is $\Z$
partially explains the integer charge law, this fact itself is not
an expression of this physical law, since $Q$ is not a local observable.

Similarly, there can be very little doubt that the \termi{charge conservation law}
is an physical/empirical law which is verifiable by (local) experimental
procedures. Although we can partially explain the charge conservation
law by the total charge operator $Q$ with the superselection rule
w.r.t.~$Q$, this is not an description of the charge conservation
law, since $Q$ is not a local observable.

\subsection{Current observables}

\label{subsec:Current-observables}

Let $\R^{1,3}$ denote the Minkowski space, and $j=(j^{0},...,j^{3}):\R^{1,3}\to\R^{1,3}$
the classical 4-electric current. If there exists a quantization $\hat{j}=(\hat{j}^{0},...,\hat{j}^{3})$
of $j$, it should be defined as an operator-valued distribution,
and each operator $\hat{j}^{\mu}(f)$ (here $f$ is a test function,
e.g., a compactly supported smooth function on $\R^{1,3}$) is expected
to be interpretable as a sort of the ``local charge observable''.%
{} The classical charge conservation law is expressed by $\di_{\mu}j^{\mu}=0$,
and so one may expect that the quantum version of this law is also
expressed by $\di_{\mu}\hat{j}^{\mu}=0$ (in the sense of operator-valued
distribution) in QFT.

In the following, the quantized current is denoted simply by $j$,
instead of $\hat{j}$. For example, for the (second quantized) Dirac
field $\psi(x)$, usually it is formally expressed as $j^{\mu}(x):=e:\ol{\psi}(x)\gamma^{\mu}\psi(x):$,
where $\gamma^{\mu}$ denotes the Dirac gamma matrices, and the colons
mean normal ordering (Wick ordering). However, note that both the
mathematically rigorous definition and the empirical meaning of $j^{\mu}$
are far from clear in general. As to the mathematical side, there
are several no-go-type theorems on the rigorous definability of $Q$
and $j^{\mu}$ (e.g., \cite{Swi1967,FPS1974}), while also some positive
results are found (e.g., \cite{Mai1972,Req1976,BDMRS1997,MS2003}).

On the other hand, in $(1+1)$-dimensional Minkowski spacetime $\R^{1,1}$,
the \termi{current algebra} (roughly, the algebra generated by the
2-current $j^{\mu}$, $\mu=0,1$), has been extensively and successfully
investigated (e.g., \cite{CHO1983,CR1987}). One may expect that this
will serve as a good test case for the study of more realistic current
algebras in $(3+1)$-dimensions. %
However, it should be noted that it is very difficult to generalize
the methods employed in the studies of the current algebras in $(1+1)$-dimensions
to higher dimensions (see \cite{Lang1994} and references therein).
Also let us quote the following (the last paragraph of \cite[Sec.6.1]{CR1987}):
\begin{quote}
For space-time dimension $d>2$, smearing the free currents at time
zero does not lead to operators (the HS condition is violated), so
that there appears to be no obvious way to make analytical sense of
such equal-time current commutation relations (in contrast to the
commutation relations for the fields, since smearing the free time-zero
Dirac fields gives rise to bounded operators for any $d$).
\end{quote}
Furthermore, it should be noted that there are some difficulties on
the ``current-oriented approach'' to QFT, even in $(1+1)$-dimensions.
Let us quote \cite[Sec.1]{CH1981}:
\begin{quote}
{[}In the Schwinger model{]} the dynamics can be expressed solely
in terms of currents without reference to the fermion fields themselves.
This possibility does not seem to be open in (QED)$_{2}$ {[}i.e.,
$(1+1)$-dimensional QED with massive fermions{]} and there more subtle
methods seem to be required. ({[}...{]} is added by H.Y.)
\end{quote}
For a function $f:\R\to\R$ and $t\in\R$, the \termi{sharp-time smeared charge operator}
$Q_{t}(f)$ at time $t$ is formally given by
\begin{equation}
Q_{t}(f)\equiv j^{0}(t,f):=\int_{-\infty}^{\infty}f(x)j^{0}(t,x)\d x,\label{eq:def:Qt}
\end{equation}
if this integral can be defined in some sense. For $S\subset\R$,
let $1_{S}(x):=1$ if $x\in S$, $1_{S}(x):=0$ otherwise. If (\ref{eq:def:Qt})
can be defined for $f=1_{S}$ for some subset $S$ of $\R$, it should
be interpreted as the charge observable in the space region $S$ at
time $t$; If $S=\R$, naturally one may expect that $Q_{t}(1_{S})$
equals the total charge operator $Q$. However it does not seem that
$Q_{t}(1_{S})$ can be defined reasonably by (\ref{eq:def:Qt}), for
a nontrivial region $S\subset\R$. In fact, even in the seemingly
trivial case where $S=\R$, the meaning of the r.h.s.~of (\ref{eq:def:Qt})
is far from clear in general. (For $(3+1)$-dimensional cases, see
\cite{Mai1972,FPS1974} and references therein. See also \cite{HR1977,Tum2022,RT2024}.)

If $Q_{t}(\cdot)$ is seen as an operator-valued Schwartz (resp.~tempered)
distribution, one should assume $f\in C_{\cpt}^{\infty}(\R)$ (compactly
supported smooth functions) (resp.~$f\in\cS(\R)$, smooth functions
of rapid decrease). In this case, if $f(x)\ge0$ and $\int_{\R}f(x)\d x=1$,
one may interpret $f$ as a ``fuzzy region'' in $\R$. For the Dirac
fields, Carey\,\&\,Hurst\,\&\,O'Brien \cite{CHO1983} adopted
the assumption $f\in\cS(\R)$, and Carey\,\&\,Ruijsenaars \cite{CR1987}
adopted the weaker assumption $f\in H_{1}(\R)$ (the Sobolev space
which consists of all absolutely continuous $L^{2}$-functions with
$L^{2}$-derivatives). They showed that (\ref{eq:def:Qt}) is reasonably
defined as a unbounded selfadjoint operator on a Hilbert space (see
Sec.\,\ref{sec:Current(1+1)}). Similarly, the sharp-time (spacial)
current operator $J_{t}$ can be defined by 
\begin{equation}
J_{t}(f)\equiv j^{1}(t,f):=\int_{-\infty}^{\infty}f(x)j^{1}(t,x)\d x,\qquad f\in H_{1}(\R).\label{eq:def:Jt}
\end{equation}
$Q_{t}$ and $J_{t}$ are shown to satisfy the canonical commutation
relation $[Q_{t}(f),J_{t}(g)]=\im B(f,g)$, where $B(\cdot,\cdot)$
is a bilinear form on $\cS(\R)$ \cite[Sec.3]{CHO1983}, that is,
they behave like boson fields (boson-fermion correspondence in $(1+1)$-dimensions).
Hence we find that (1) both $Q_{t}(f)$ and $J_{t}(g)$ ($f,g\not\equiv0$)
have the continuous spectrum $\R$, and (2) $Q_{t}$ and $J_{t}$
(i.e., $j^{0}$ and $j^{1}$) do not commute in general. These facts
raise the following questions.
\begin{question}
\label{que:integer}How can we express the integer charge law in terms
of $Q_{t}$ and $J_{t}$ (or $j^{0}$ and $j^{1}$)? How do the measurements
of $\R$-valued observables $Q_{t}$ and $J_{t}$ lead to the integer
charge law? How are such measurements related to e.g., Millikan's
oil drop experiments?
\end{question}

\begin{question}
\label{que:conservation}If we cannot measure the charge density $j^{0}$
and the spacial current $j^{1}$ simultaneously, then how can we justify
to call the equation $\di_{\mu}j^{\mu}=0$ the ``charge conservation
law''? How can we verify the equation $\di_{\mu}j^{\mu}=0$ experimentally? 
\end{question}

In these questions, I assumed that the current $j^{\mu}$ can be rigorously
defined as an operator-valued distribution in $(1+1)$-dimensions.
In $(3+1)$-dimensions, the situation would get even worse in that
the rigorous definability of 4-current $j^{\mu}$ is not be made clear
yet. Of course we are all familiar with both the integer charge law
and the charge conservation law, and (perhaps) believe that they can
be verified experimentally. Nevertheless, it seems that relativistic
QFT has not even been able to give any precise statements of these
empirical laws; that is, perhaps no clear answer has yet been found
to the simple questions ``what is the integer charge law?'' and
``what is the charge conservation law?''.

I am somewhat skeptical of the prospects on the existence of the positive
answers to Questions \ref{que:integer}, \ref{que:conservation}.
Instead I would like to raise the following
\begin{question}
Can we give the precise statements of these two empirical laws in
terms of local observables, but without employing the notion of 4-current
(or 2-current)?
\end{question}

As I mentioned above, it is unlikely that we can define the observable
$Q(1_{S})$ expressing the charge in a nontrivial spacial region $S$.
If instead we consider a ``fuzzy region $f$'', i.e., a real-valued
test function $f\in C_{\cpt}^{\infty}(\R^{d})$ ($d$: space dimension),
then we can define the smeared operator $Q(f)$ when $d=1$, which
might be interpreted as the charge within the fuzzy region $f$. However,
this interpretation is fragile because of Questions \ref{que:integer}
and \ref{que:conservation}, and further the definability of $Q(f)$
is not clear for $d>1$. Therefore it seems that there is not much
hope to justify the notion of ``the (total) charge within a (possibly
fuzzy) spacial region''. Instead I would like to consider the notion
of ``the (total) charge in a \emph{scope}'', similarly to the notion
of ``the photon number in a \emph{scope}''. 

\subsection{\protect\label{subsec:particle}Relation to the spatial localizability
of a particle}

Evidently, the problem of the (spacial) localizability of charge is
related to that of localizability of a particle, although the former
is not the same as the latter (see e.g., \cite{Tum2022}). Usually
we believe that the charge of a positron/electron is detectable \emph{locally},
e.g., as a trajectory in a cloud (or bubble) chamber. However notice
that it is likely that the concept of (spatially localizable) particle
is no longer maintainable in relativistic QFT. Let us quote Halvorson\,\&\,Clifton
\cite[Sec.1]{HC2002};
\begin{quote}
It is a widespread belief, at least within the physics community,
that there is no particle mechanics that is simultaneously relativistic
and quantum-theoretic; and, thus, that the only relativistic quantum
theory is a \emph{field} theory. This belief has received much support
in recent years in the form of rigorous ``no-go theorems'' by Malament
\cite{Mal1996} and Hegerfeldt \cite{Hag1998a,Hag1998b}.
\end{quote}
I agree with this view. On the other hand, I do not fully agree with
the following \cite[Sec.7]{HC2002}:
\begin{quote}
The argument for localizable particles appears to be very simple:
Our experience shows us that objects (particles) occupy finite regions
of space. But the reply to this argument is just as simple: These
experiences are illusory!
\end{quote}
However, for example, the trajectories seen in a cloud/bubble chamber
will not be an illusion or an artifact. (If so, all the discoveries
in %
particle physics (e.g., of neutrons, muons, pions, etc.) obtained
with cloud/bubble chambers should be called ``illusions''!) Hence
at least physicists should be able to explain and predict such phenomena
in a cloud chamber, in terms of some mathematical formulation, rigorously
but simply as well as possible.\footnote{For example, an ``explanation'' in terms of standard QED (with renormalization
procedures) seems neither rigorous nor simple. In fact, ultimately
QED itself should be explained in a suitable manner. Perhaps, a simplified,
schematic, but rigorous (logically firm and consistent) explanation
will be conceptually more significant. Also note that a ``simple
but non-rigorous explanation'' tends to be based on one's impressions,
and so tends to lack generality.}

Even if we give up understanding a trajectory in a cloud chamber as
(an appearance of) a localized particle, there is a possibility of
the interpretation that the trajectory is (an appearance of) a localized
\emph{charge}. However, even the latter interpretation does not seem
possible without any restriction, as I mentioned in Subsec.\,\ref{subsec:Current-observables},
see also \cite{Tum2022,RT2024}. 

Again consider the example of photon number observable. Probably it
is impossible to give any meaning to ``the number of all photons
in a space region $\cO$'', in general. However, in quantum optics,
it appears that one can consistently refer to the photon number in
a pulse of laser beam, which is spatially localized. I would like
to express this fact as follows: The latter photon number observable
can be understood as the total number of photons not only in the space
region $\cO$, but also in a restricted ``scope'' (so-called ``mode'',
in quantum optics). Let us call this the concept of \termi{scope-localized particles}.
However this is not the main topic of this article; instead I will
examine the concept of \termi{scope-localized charge} here.

\section{CAR algebra, Araki\textendash Wyss map}

\label{sec:Araki-Wyss-map}

In this paper we confine ourselves to charged fermion fields. Furthermore
the existence of the sharp-time field is assumed. This assumption
is satisfied by all free fields, and also by some interacting fields
with cutoff in some sense (e.g., Hamiltonian lattice field theories).
In this case, each sharp-time field (especially the time-zero field)
is expressed as a CAR $C^{*}$-algebra defined as follows.

Let $(\cV,\langle\cdot|\cdot\rangle)$ be a complex inner product
space, and $\ol{\cV}$ the completion of $\cV$, which is a Hilbert
space. For $f\in\ol{\cV}$, let $f^{*}\in\cV^{*}$ denote the functional
$f^{*}:g\mapsto\langle f|g\rangle$, $g\in\ol{\cV}$. 

The (charged) \termi{CAR $C^{*}$-algebra} $\SYM{\CAR(\cV)}{CAR}$
\cite{BR97} is the $C^{*}$-algebra generated algebra ``freely''
generated by the elements $\{\Psi(f)|f\in\cV\}$ satisfying 
\[
\{\Psi(f),\Psi(g)\}=0,\qquad\{\Psi(f),\Psi^{*}(g)\}=\langle f|g\rangle\bbOne.
\]
Typically, $\cV$ is set to the test function space $C_{\cpt}^{\infty}(\R^{d})$
($d$: space dimension) of the time-zero field, which is a dense subspace
of $L^{2}(\R^{d})$. In this case, all generators $\Psi(f),\Psi^{*}(f)$
($f\in\cV$) are called ``local'' in the sense that $\supp(f)$
is compact for all $f\in\cV$, and that
\[
\supp(f)\cap\supp(g)=\emptyset\quad\Then\quad\{\Psi(f),\Psi(g)\}=\{\Psi(f),\Psi^{*}(g)\}=0.
\]
For any $t\in\R$, the map $f\mapsto e^{\im t}f$ on $\cV$ induces
an automorphism $\tau_{t}$ of $\CAR(\cV)$. The automorphism group
$\{\tau_{t}|t\in\R\}\cong U(1)$ is called the (global) \termi{gauge group},
and each $\tau_{t}$ is called a \termi{gauge transformation}. Usually,
the \termi{observables} in $\CAR(\cV)$ are identified with the gauge-invariant
elements of $\CAR(\cV)$. For example, any element of $\CAR(\cV)$
of the form $\Psi^{*}(f)\Psi(f)$ ($f\in\cV$) is seen as a \emph{local}
observable.

Let $K:\cV\to\cV$ be a finite-rank linear operator on $\cV$. %
{} That is, there exists $n\in\N$ and an orthonormal system $\{e_{i}\}_{i=1,...,n}$
of $\cV$ such that $K=\sum_{i,j=1}^{n}K_{ij}e_{i}\otimes e_{j}^{*}$,
$K_{ij}:=\langle e_{i}|Ke_{j}\rangle$, equivalently $Kf=\sum_{i,j=1}^{n}K_{ij}\langle e_{j}|f\rangle e_{i}$
for all $f\in\cV$. Note that $K^{*}=\sum_{i,j=1}^{n}\ol{K_{ij}}e_{j}\otimes e_{i}^{*}$.
{} %
Let
\[
\SYM{\d\Gamma(K)}{dGa}\equiv K\Psi^{*}\Psi:=\sum_{i,j=1}^{n}K_{ij}\Psi^{*}(e_{i})\Psi(e_{j})=\sum_{j=1}^{n}\Psi^{*}(Ke_{j})\Psi(e_{j}).
\]
This is independent of the choice of the orthonormal system $\{e_{i}\}$.
Let $\SYM{\FR(\cV)}{FR}$ denote the {*}-algebra of finite-rank operators
on $\cV$. Then we have the linear map $\d\Gamma:\FR(\cV)\to\CAR(\cV)$,
$K\mapsto K\Psi^{*}\Psi$. Let us call $\d\Gamma$ the \termi{Araki--Wyss map}
(\cite{AW1964}, see also \cite{CHO1983,CR1987,Tha1992}). We see
$\d\Gamma(K)$ is an observable, i.e., gauge-invariant. If all generators
of $\CAR(\cV)$ are local, $\d\Gamma(K)$ is a \emph{local} observable.

For any $K\in\FR(\cV)$, there are two orthonormal systems $\{f_{i}\}$
and $\{g_{i}\}$ in $\cV$ such that
\[
K=\sum_{i=1}^{n}\lambda_{i}f_{i}\otimes g_{i}^{*},\qquad\lambda_{i}\ge0,
\]
and hence
\[
\|K\Psi^{*}\Psi\|\le\sum_{i=1}^{n}\lambda_{i}=\|K\|_{1}\qquad(\text{trace norm}).
\]
Thus, if $\FR(\cV)$ is given the trace norm, the map $\d\Gamma:\FR(\cV)\to\CAR(\cV)$
is continuously extended for trace-class operators $K$ on $\cV$,
as pointed out by Araki and Wyss \cite{AW1964}. %

Note that $\d\Gamma$ is a Lie algebra homomorphism:
\[
[\d\Gamma(A),\d\Gamma(B)]=\d\Gamma([A,B]),\qquad\forall A,B\in\FR(\cV),
\]
and furthermore $\d\Gamma$ satisfies
\[
[\d\Gamma(A),\Psi^{*}(f)]=\Psi^{*}(Af),\qquad[\d\Gamma(A),\Psi(f)]=-\Psi(A^{*}f),\qquad\forall A\in\FR(\cV),f\in\cV,
\]
and hence
\[
\Psi(e^{A}f)=e^{\d\Gamma(A)}\Psi(f)e^{-\d\Gamma(A)},\qquad\forall f\in\cV.
\]
Let $\SYM{\FRP(\cV)}{FRP}$ denote the set of (orthogonal) projections
on $\cV$ of finite rank. 

Recall that for a unital $C^{*}$-algebra $\fA$ and $x\in\fA$, the
spectrum $\Spec_{\fA}(x)$ of $x$ in $\fA$ is defined by
\[
\Spec_{\fA}(x):=\{\lambda\in\C|(x-\lambda\bbOne)\text{ is not invertible in }\fA\},
\]
and also the following simple fact:
\begin{lem}
\label{lem:E1EnSpec-1}Let $E_{1},...,E_{n}$ be nonzero projections
in a unital $C^{*}$-algebra $\fA$ such that $i\neq j\then E_{i}E_{j}=0$
and $\sum_{i=1}^{n}E_{i}=\bbOne$. Let $A:=\sum_{i=1}^{n}z_{i}E_{i}$,
$z_{i}\in\C$. Then $\Spec_{\fA}(A)=\{z_{1},...,z_{n}\}$.%
\end{lem}

From the above lemma we readily find the following
\begin{lem}
For any $P\in\FRP(\cV)$, the spectrum of $\Gamma(P)$ in $\fA:=\CAR(\cV)$
is given by $\Spec_{\fA}(\Gamma(P))=\{0,1,...,n_{P}\}$, where $n_{P}:={\rm rank}(P)=\Tr P$.
\end{lem}

Let $\cfunc:\FR(\cV)\to\C$ be a linear functional. Let

\begin{equation}
\SYM{\Q_{\cfunc}(A)}{Qal}:=\d\Gamma(A)-\cfunc(A)\bbOne.\label{eq:def:Qalpha}
\end{equation}
Then we see
\[
[\Q_{\cfunc}(A),\Q_{\cfunc}(B)]=[\d\Gamma(A),\d\Gamma(B)]=\d\Gamma([A,B])=\Q_{\cfunc}([A,B])+\cfunc([A,B])\bbOne,
\]
and hence $\Q_{\cfunc}:\FR(\cV)\to\CAR(\cV)$ is not a Lie algebra
homomorphism, in general. However, we have
\begin{equation}
[\Q_{\cfunc}(A),\Q_{\cfunc}(B)]=0\quad\text{ whenever }\quad[A,B]=0.\label{eq:=00005BA,B=00005D=00003D0}
\end{equation}
Especially, if $P,P'\in\FRP(\cV)$ and $P'\le P$, then $[\Q_{\cfunc}(P),\Q_{\cfunc}(P')]=0$.

\def\algorithmname{Working Hypothesis}

I propose the following
\begin{lyxalgorithm}
\label{wh:Q}Any ``scope'' is mathematically expressed by a finite-rank
projection operator $P$ on $\cV$, i.e., $P\in\FRP(\cV)$. %
{} Let $\SYM{\Scope(\cV)}{Scope}\subset\FRP(\cV)$ denote the set of
``scopes''. For any $P\in\Scope(\cV)$, the local observable of
``the total charge in scope $P$'' is expressed as $\Q_{\cfunc}(P)$,
for some fixed functional $\cfunc:\FR(\cV)\to\C$. We assume that
the spectrum of $\Q_{\cfunc}(P)$ is a (finite) subset of $\Z$ for
all $P\in\Scope(\cV)$. This can be interpreted as an expression of
the integer charge law.

\end{lyxalgorithm}

More concretely, I would like to put the following
\begin{lyxalgorithm}
\label{wh:even}(1) $\cfunc$ is given by $\cfunc(A):=\frac{1}{2}\Tr A$,
and (2) every $P\in\Scope(\cV)$ is of even-rank, i.e., $n_{P}=\Tr P$
is even.
\end{lyxalgorithm}

I suppose that the concept of ``scope'' should have some empirical
meaning, and also the set $\Scope(\cV)$ of scopes should have not
only a purely mathematical meaning; Some mathematical definitions
of $\Scope(\cV)$ will be ``good definitions'' with respect to that
empirical meaning, but others will not. However, currently I have
no clear explanation of the empirical meaning of ``scope'', and
hence I cannot specify what definition of $\Scope(\cV)$ is a good
one. Thus I will search for ``good requirements'' for $\Scope(\cV)$
in a somewhat heuristic way, giving several working hypotheses such
as \ref{wh:Q}, \ref{wh:even}.

By Working Hypothesis \ref{wh:even}, for any $P\in\Scope(\cV)$,
the spectrum of $\Q_{\cfunc}(P)$ is the integers $\{k\in\Z:\,|k|\le n_{P}/2\}$.
However note that even if $n_{P}$ is odd, the local observable $\Q_{\cfunc}(P)$
is defined. In this case, the spectrum of $\Q_{\cfunc}(P)$ is the
set of half integers $\{k\in\Z+\frac{1}{2}:\,|k|\le n_{P}/2\}$, which
cannot be interpreted as the values of charge in the usual sense.
Probably the observable $\Q_{\cfunc}(A)$ has some physical/empirical
significance for every $A\in\FR(\cV)_{{\rm sa}}$, but there are many
cases where its meaning %
is not immediately clear. Similarly, every $P\in\FRP(\cV)$ might
have some significance as a ``scope'', whether $n_{P}$ is even
or odd. Hence it is desirable that Working Hypothesis \ref{wh:even}
is reasonably justified mathematically and/or empirically. Unfortunately
currently I do not have enough justification for Working Hypothesis
\ref{wh:even}; I will simply hypothesize it for the better interpretability
of $\Q_{\cfunc}(P)$. An attempt to justify Working Hypothesis \ref{wh:even}
will be given in Subsec.\,\ref{subsec:Charge-conjugation}.

{} %

For each $P\in\Scope(\cV)$, let $\SYM{\fA_{P}}{AP}$ denote the (finite-dimensional)
$C^{*}$-subalgebra of $\CAR(\cV)$ generated by $\bbOne$ and $\{\Psi^{*}(f)\Psi(g)|f,g\in\ran(P)\}$.
Let us call $\fA_{P}$ the \termi{algebra of observables in the scope
$P$}.

It is easy to see the following
\begin{lem}
Let $P\in\FRP(\cV)$. Then $[\Psi^{*}(f)\Psi(g),\d\Gamma(P)]=0$ for
all $f,g\in\ran(P)$.
\end{lem}

Hence we also have
\begin{lem}
$\Q_{\cfunc}(P)$ is in the center of $\fA_{P}$, i.e., $[\Q_{\cfunc}(P),X]=0$
for all $X\in\fA_{P}$.
\end{lem}

This can be understood as one of the rigorous expressions of the charge
conservation law, in terms of local observables. Intuitively I would
like to interpret this fact as follows: If $X$ is a (local) observable
in the scope $P$, then the measurement of $X$ cannot change the
value of the charge $\Q_{\cfunc}(P)$ in the scope $P$; Slightly
more generally, any local \emph{operation} in the scope $P$ cannot
change the value of $\Q_{\cfunc}(P)$.

\section{Quasi-free states, Lundberg's extensions}

If we fix a representation of $\CAR(\cV)$ on a Hilbert space $\cH$,
sometimes $\Q_{\cfunc}(A)$ is extended to some non-trace-class operators
$A$. This fact is found by Lundberg \cite{Lun1976}, and is employed
by Carey et al.~\cite{CHO1983,CR1987} to formulate the current algebras
in $(1+1)$ dimensions.

Assume that $\cV$ is a Hilbert space. Let $B(\cV)$ denote the algebra
of the bounded operators on $\cV$, and $B(\cV)_{{\rm sa}}\subset B(\cV)$
the subset of selfadjoint operators. First notice the following
\begin{lem}
\label{lem:alphaIP}Let $\cI$ be an involution on $\cV$. (Recall
that an antilinear operator $\cI$ on $\cV$ is called an \termi{involution}
if $\cI^{2}=\bOne$.) Let $P$ be an (orthogonal) projection on $\cV$
such that $\cI P=P\cI$. Then there exists a unique automorphism $\alpha_{\cI,P}$
of $\CAR(\cV)$ such that
\[
\alpha_{\cI,P}(\Psi(f))=\Psi((\bOne-P)f)+\Psi^{*}(\cI Pf).
\]
\end{lem}

\begin{proof}
Set $U:=\bOne-P$ and $V:=\cI P$, then we see%
\[
V^{*}U=VU^{*}=0,\qquad U^{*}U+V^{*}V=\bOne=UU^{*}+VV^{*}.
\]
(Recall that for an antilinear operator $A$, its adjoint $A^{*}$
is defined by $\langle u|Av\rangle=\ol{\langle A^{*}u|v\rangle}\thinspace(=\langle v|A^{*}u\rangle)$.)
Hence the lemma follows from \cite[Theorem 5.2.5]{BR97}. %
\end{proof}
Let $T\in B(\cV)_{\mathrm{sa}}$ (``sa'' denotes ``selfadjoint'')
such that $0\le T\le\bOne$, where $\bOne$ denotes the identity in
$B(\cV)$. A \termi{gauge-invariant quasi-free state} $\omega_{T}$
of $\CAR(\cV)$ is uniquely defined by the $n$-point functions (cf.~\cite[Sec.17.2.3]{DG2022})
\[
\omega_{T}(\Psi(f_{n})^{*}\cdots\Psi(f_{1})^{*}\Psi(g_{1})\ldots\Psi(g_{m}))=\begin{cases}
0 & \text{if }n\neq m\\
\det(\langle g_{i}|Tf_{j}\rangle)_{i,j=1,...,n} & \text{if }n=m
\end{cases}
\]
We see $\omega_{T}(\d\Gamma(A))=\tr(TA).$

Let $(\SYM{\mathcal{H}_{T}}{HT},\SYM{\pi_{T}}{piT},\SYM{\Omega_{T}}{OmegaT})$
be the GNS representation of $\CAR(\cV)$ w.r.t.~the state $\omega_{T}$.
Let $\SYM{\Psi_{T}}{aT}(f):=\pi_{T}(\Psi(f))$.
\begin{example}
When $T=0$, the GNS representation $(\cH_{0},\pi_{0},\Omega_{0})$
is equivalent to the standard Fock representation of $\CAR(\cH)$
on the fermion Fock space $\cF(\cV)$ over $\cV$, so that $\Psi_{0}(f)\Omega_{0}=0$
for all $f\in\cV$.

Next consider the case where $T$ is a projection $P$ on $\cV$.
Let $\cI$ be an involution on $\cV$ such that $\cI P=P\cI$. Then,
by Lemma \ref{lem:alphaIP}, the operator $\Psi_{P}(\cdot)$ can be
related to $\Psi_{0}(\cdot)$ by
\[
\Psi_{P}(f)=\Psi_{0}^{*}(\cI Pf)+\Psi_{0}((\bOne-P)f),\qquad f\in\cV.
\]
 with $\Omega_{P}=\Omega_{0}$.%
{} \hfill $\Box$
\end{example}

{}

For $A\in\FR_{{\rm sa}}(\cV)$, let
\[
\cfunc_{T}(A):=\omega_{T}(\d\Gamma(A))
\]
\[
\SYM{\Q_{T}(A)}{QT}:=\pi_{T}(\d\Gamma(A))-\cfunc_{T}(A)\bbOne
\]
Then we have \cite[Sec.1]{Lun1976}
\begin{align}
\im[\Q_{T}(A),\Q_{T}(B)] & =\Q_{T}(\im[A,B])-2\Im\tr(TA(\bOne-T)B)\bbOne.\label{eq:=00005BQ(A),Q(B)=00005D}
\end{align}

\[
\SYM{W_{T}(A)}{WT}=e^{\im\cdot\Q_{T}(A)},\qquad A\in\FR(\cV)_{{\rm sa}}
\]
Let
\[
\SYM{O_{T}(\cV)}{OT}=\{A\in B(\cV)_{\mathrm{sa}};\tr(TA(\bOne-T)A)<\infty\}.
\]
The following theorem extends the maps $\Q_{T}$ and $W_{T}$ defined
on $\FR(\cV)_{{\rm sa}}$ to the maps $\tilde{\Q}_{T}$ and $\tilde{W}_{T}$
defined on $O_{T}(\cV)$, respectively (Lundberg \cite{Lun1976},
see also \cite{CHO1983,CR1987,Tha1992}):
\begin{thm}[Lundberg]
\label{thm:Lundberg}There exists a unique map $\tilde{W}_{T}:O_{T}(\cV)\to\pi_{T}(\CAR(\cV))''$
such that for all $A\in O_{T}(\cV)$,
\begin{enumerate}
\item $\{\tilde{W}_{T}(sA)|s\in\R\}$ is a strongly continuous unitary one-parameter
group,
\item $\Psi_{T}(e^{\im sA}f)=\tilde{W}_{T}(sA)\Psi_{T}(f)\tilde{W}_{T}(sA)^{-1},\qquad\forall f\in\cV,s\in\R,$
\item Let $\SYM{\tilde{\Q}_{T}(A)}{QTtil}$ denote the (possibly unbounded)
selfadjoint operator on $\cH_{T}$ such that $\tilde{W}_{T}(sA)=e^{\im s\tilde{\Q}_{T}(A)}$
($s\in\R$).%
{} Then $\Omega_{T}\in D(\tilde{\Q}_{T}(A))$, 
\item $\left\langle \Omega_{T}\big|\tilde{\Q}_{T}(A)\Omega_{T}\right\rangle =0$
\end{enumerate}
For $A,B\in O_{T}(\mathcal{H})$ the following identity holds
\[
\tilde{W}_{T}(A)\tilde{W}_{T}(B)\tilde{W}_{T}(A)^{-1}=\tilde{W}_{T}(e^{\im A}Be^{-\im A})e^{\im b_{T}(A,B)},
\]
where
\begin{equation}
\SYM{b_{T}}{bT}(A,B):=-2\Im\int_{0}^{1}\tr(TA(\bOne-T)e^{\im sA}Be^{-\im sA})\d s.\label{eq:bT}
\end{equation}
\end{thm}

Moreover, Eq.\,(\ref{eq:=00005BQ(A),Q(B)=00005D}) can be generalized
for $\tilde{\Q}_{T}$ \cite{Lun1976}. The last term $-2\Im\tr(TA(\bOne-T)B)$
in (\ref{eq:=00005BQ(A),Q(B)=00005D}) (replacing $\Q_{T}$ with $\tilde{\Q}_{T}$)%
{} is called the \termi{Schwinger term}. However notice the following
simple fact:
\begin{lem}
\label{lem:TAB}Let $T,A,B$ be bounded selfadjoint operators on a
Hilbert space $\cH$. Assume $[A,B]=0$.
\begin{enumerate}
\item If $A$ or $B$ is trace-class, then $\tr(TAB)\in\R,\ \tr(TATB)\in\R$,
and hence $\tr(TA(\bOne-T)B)\in\R$.
\item When neither $TAB$ nor $TATB$ is trace-class%
, there are some cases where $TA(\bOne-T)B$ is trace-class, but $\tr(TA(\bOne-T)B)\notin\R$.
\end{enumerate}
\end{lem}

\begin{proof}
(1) is obvious. An example of (2) shall be given in Section \ref{sec:Current(1+1)}.
\end{proof}
Therefore (\ref{eq:=00005BA,B=00005D=00003D0}) does not hold for
$\tilde{\Q}_{T}$. That is, we find a rather surprising phenomenon
that even if $A,B\in O_{T}(\cV)$ satisfy $AB=BA$, $\tilde{\Q}_{T}(A)$
and $\tilde{\Q}_{T}(B)$ may not commute. In other words, Lundberg's
extension $\tilde{\Q}_{T}$ was made possible at the cost of the commutativity
(\ref{eq:=00005BA,B=00005D=00003D0}). Algebraically, this amounts
to a central extension of an algebraic structure (of group or Lie
algebra). Eq.\,(\ref{eq:bT}) implies that if $AB=BA$,
\begin{equation}
b_{T}(A,B)=-2\Im\tr(TA(\bOne-T)B),\label{eq:bt-commute}
\end{equation}
which equals the Schwinger term. Lemma \ref{lem:TAB}(2) says that
there are some cases where $AB=BA$ but $b_{T}(A,B)\neq0$.

\section{Current algebra in $(1+1)$-dimensions}

\label{sec:Current(1+1)}

A physically significant example of Lemma \ref{lem:TAB}(2) is given
in the $(1+1)$-dimensional Dirac field as follows \cite[Sec.3]{CHO1983}.
Let $\gamma^{\mu}$ ($\mu=0,1$) be Dirac gamma matrices in $(1+1)$-dimensions,
for example,
\[
\gamma^{0}:=\begin{pmatrix}0 & 1\\
1 & 0
\end{pmatrix},\qquad\gamma^{1}:=\begin{pmatrix}0 & -1\\
1 & 0
\end{pmatrix},\qquad\text{with }\gamma^{5}:=\gamma^{0}\gamma^{1}=\begin{pmatrix}1 & 0\\
0 & -1
\end{pmatrix}.
\]
Let $\cV:=L^{2}(\R,\C^{2})\cong L^{2}(\R)\oplus L^{2}(\R)$, and consider
$\CAR(\cV)$.%
{} Let $m\ge0$. Define the Dirac Hamiltonian (without second quantization)
on $\cV$ by
\begin{equation}
H:=-\im\hbar c\gamma^{5}\frac{\d}{\d x}+mc^{2}\gamma^{0}\quad\text{ or }\quad H:=-\im\gamma^{5}\frac{\d}{\d x}+m\gamma^{0}\text{ if we set }c=\hbar=1.\label{eq:2D-Hamiltonian}
\end{equation}
In the following we set $c=\hbar=1$ unless otherwise stated. The
corresponding positive (resp.~negative) energy projection $E_{+}$
(resp.~$E_{-}$) is defined by 
\begin{equation}
\SYM{E_{\pm}}{E+-}:=\frac{1}{2}\left(\bOne\pm\frac{H}{|H|}\right).\label{eq:def:E+-}
\end{equation}
For $f=(f_{0},f_{1})\in\cS(\R,\R^{2})$ ($\R^{2}$-valued smooth functions
of rapid decrease), define the multiplication operator $\cJ(f)$ by
\[
(\SYM{\cJ(f)}{J(f)}g)(x):=(f_{0}(x)+\gamma^{5}f_{1}(x))g(x),\qquad g\in L^{2}(\mathbb{R},\mathbb{C}^{2}).
\]
Then $T:=E_{-}$, $A:=\cJ(f)$ and $B:=\cJ(g)$ ($f,g\in\cS(\R,\R^{2})$)
give an example Lemma \ref{lem:TAB}(2), when $f$ and $g$ satisfy
some suitable conditions ($f,g\not\equiv0$, $f\neq g$, $\supp(f)\cap\supp(g)\neq\emptyset$,
etc). The GNS representation $(\cH_{T},\pi_{T},\Omega_{T})$ of $\CAR(\cV)$
is equivalent to the usual Fock representation of the (second quantized)
Dirac field.

The selfadjoint operator $\tilde{\Q}_{E_{-}}(\cJ(f))$ is shown \cite[Sec.5]{CHO1983}
to be identified with the smeared current operator $j(f)$ formally
defined as 
\begin{equation}
j(f):=j^{0}(f_{0})+j^{1}(f_{1}),\qquad j^{\mu}(f_{\mu}):={\displaystyle \int_{-\infty}^{\infty}}\wick{\overline{\Psi}(x)\gamma^{\mu}\Psi(x)}f_{\mu}(x)\d x,\qquad\mu=0,1.\label{eq:1.1}
\end{equation}
Therefore Lundberg's Theorem \ref{thm:Lundberg} leads to a rigorous
formulation of the current algebras in $(1+1)$-dimensions. 

\paragraph{The disadvantage of the approach of Lundberg and Carey et al.~to
current algebras.}

Thus we reviewed the approach to current algebras in $(1+1)$-dimensions
\cite{CHO1983,CR1987} based on Lundberg's Theorem \ref{thm:Lundberg}.
However, we find that $j^{0}$ and $j^{1}$ satisfy the CCR $[j^{0}(f),j^{1}(g)]=\im B(f,g)$,
where $B$ is a bilinear form, which is nonzero due to the Schwinger
term \cite[Sec.5]{CHO1983}. Thus we are forced to be puzzled by Questions
\ref{que:integer}, \ref{que:conservation}. Furthermore, it does
not seem that we can also apply Theorem \ref{thm:Lundberg} to the
definition the current algebras in $(3+1)$-dimensions.

Viewed from a broader perspective, the fact that the Lundberg extension
$\tilde{\Q}_{\cfunc}$ is representation-dependent makes it difficult
for us to generalize this approach to QFT in curved spacetime (QFTCS),
because the mainstream of QFTCS appears to tend to the representation-independent
(e.g., abstract $C^{*}$-algebraic) approach, see e.g., \cite{Dim1982,Wal1994,BFV2003,BGP2007,BF2009,HW2015,KM2015,Ger2019}.
Furthermore, since Lundberg's Theorem \ref{thm:Lundberg} relies upon
quasi-free states, it is not clear whether it can be generalized to
interacting (non-quasi-free) fields even in (flat) Minkowski spacetime.

Therefore, it seems that for a general theory (containing e.g., interacting
fields and/or curved spacetimes), it is a better way to consider the
scope-local charge $\Q_{\cfunc}$ based on the original Araki\textendash Wyss
map $\d\Gamma$, not the Lundberg extension $\tilde{\Q}_{\cfunc}$
of $\Q_{\cfunc}$. In the next section I will return to $\Q_{\cfunc}$. 

\section{Scope-local charge $\protect\Q_{\protect\cfunc}$}

\label{sec:Scope-local}

\subsection{Charge conjugation}

\label{subsec:Charge-conjugation}

Next we examine the problem to describe the notion of charge in terms
of local observables with Working Hypotheses \ref{wh:Q}, \ref{wh:even}
without relying upon the notion of current. We work in a charged spinor
field in $(1+1)$ or $(3+1)$-dimensions (mainly the latter). We consider
the operator $\Q_{\cfunc}(A)$ based on the original Araki\textendash Wyss
map $\d\Gamma(A)$ defined only for $A\in\FR(\cV)$ (or the trace-class
extension of it), not the Lundberg extension $\tilde{\Q}_{\cfunc}$
of $\Q_{\cfunc}$.

Let $\cV:=C_{\cpt}^{\infty}(\R^{d},\C^{4})$ ($d=1,3$) (compactly
supported $\C^{4}$-valued smooth functions on $\R^{d}$) with the
usual $L^{2}$ inner product $\langle\cdot|\cdot\rangle$, so that
its completion is $\ol{\cV}=L^{2}(\R^{d},\C^{4})$. Here, $\C^{4}$
is given the usual inner product, and hence $\langle f|g\rangle=\int_{\R^{d}}\d^{d}\x\,\sum_{k=1}^{d}\ol{f_{k}(\x)}g_{k}(\x)$,
$f,g\in\ol{\cV}$.

In $(3+1)$-dimensions, the Dirac matrices $\balpha=(\alpha_{1},\alpha_{2},\alpha_{3})$
and $\beta$ are $4\times4$ matrices satisfying
\[
\{\alpha_{i},\alpha_{k}\}=2\delta_{ik}\bbOne,\qquad\{\alpha_{i},\beta\}=0,\qquad\beta^{2}=\bbOne.\qquad i,k=1,2,3.
\]
The Dirac gamma matrices are $\gamma^{0}:=\beta$, $\gamma^{k}:=\beta\alpha_{k}$,
satisfying $\{\gamma^{\mu},\gamma^{\nu}\}=2g^{\mu\nu}\bbOne$. The
Majorana representation is given by
\[
\beta=\beta^{{\rm M}}:=\begin{pmatrix}0 & \im\bbOne_{2}\\
-\im\bbOne_{2} & 0
\end{pmatrix},\quad\alpha_{i}=\alpha_{i}^{{\rm M}}:=\begin{pmatrix}\sigma_{i} & 0\\
0 & -\sigma_{i}
\end{pmatrix}\text{ for }i=1,3,\quad\alpha_{2}=\alpha_{2}^{{\rm M}}:=\begin{pmatrix}0 & \bbOne_{2}\\
\bbOne_{2} & 0
\end{pmatrix}
\]
where $\sigma_{i}$ $(i=1,2,3)$ are the standard Pauli matrices,%
{} and $\bbOne_{2}$ is the $2\times2$ identity matrix. In this case,
any element of $\gamma^{\mu}=\gamma^{{\rm M},\mu}$ ($\mu=0,...,3$)
is purely imaginary, and the \termi{charge conjugation} $\cC=\SYM{\cC^{{\rm M}}}{CM}$
is defined to be the complex conjugation $\psi\mapsto\ol{\psi}$ on
$\C^{4}$ (warning: $\ol{\psi}$ does not refer to the Dirac adjoint
of $\psi$, in this paper). We can obtain any other representation
of the gamma matrices and the charge conjugation from the Majorana
representation by $\gamma^{\mu}=S\gamma^{{\rm M},\mu}S^{-1}$, $\cC=S\cC^{{\rm M}}S^{-1}$
for some invertible $4\times4$ matrix $S$.

In $(1+1)$-dimensions, the Majorana representation of the $2\times2$
gamma matrices are given by
\[
\beta=\gamma^{0}=-\sigma_{2}=\begin{pmatrix}0 & \im\\
-\im & 0
\end{pmatrix},\quad\alpha_{1}=\gamma^{0}\gamma^{1}=\gamma^{5}=\sigma_{1}=\begin{pmatrix}0 & 1\\
1 & 0
\end{pmatrix},\quad\gamma^{1}=\beta\alpha_{1}=\im\sigma_{3}=\begin{pmatrix}\im & 0\\
0 & -\im
\end{pmatrix},
\]
and similarly in this representation the charge conjugation $\cC$
is given by the complex conjugate on $\C^{2}$.

If we consider the ($3+1$)-dimensional (free) Dirac field, the Hamiltonian
(without second quantization) is given by
\begin{equation}
H=-\im\hbar c\balpha\cdot\nabla+\beta mc^{2}=-\im\hbar c\sum_{k=1}^{3}\gamma^{0}\gamma^{k}\di_{k}+\gamma^{0}mc^{2}.\label{eq:4D-Hamiltonian}
\end{equation}
In $(1+1)$-dimensions, it is given by (\ref{eq:2D-Hamiltonian}).
In any case, the corresponding positive (resp.~negative) energy projection
$E_{+}$ (resp.~$E_{-}$) is defined by (\ref{eq:def:E+-}).

Let $\cfunc:\FR(\cV)\to\C$ be a linear functional%
, and recall the definition (\ref{eq:def:Qalpha}) of $\Q_{\cfunc}$.
If the charge conjugation $\cC$ is understood to act on $\cV$ (or
$\ol{\cV}$) pointwise, it is an involution on $\cV$ (or $\ol{\cV}$).

In the following, we may assume the Majorana representation of gamma
matrices for simplicity, so that $\cC$ equals the complex conjugation.
Define the automorphism $\SYM{\scC}C$ on $\CAR(\cV)$ by
\[
\scC(\Psi(f)):=\Psi^{*}(\cC f)\ (=\Psi^{*}(\ol f)),\qquad f\in\cV,
\]
cf.~\cite[Theorem 5.2.5(5)]{BR97}. Let us call $\scC$ the \termi{charge conjugation automorphism}.
Note that $\scC(\Psi^{*}(f))=\scC(\Psi(f))^{*}=(\Psi^{*}(\cC f))^{*}=\Psi(\cC f)$.
\begin{lem}
\label{lem:C(Q)}For any $P\in\FRP(\cV)$, we have $\scC\left(\Q_{\cfunc}(P)\right)=-\Q_{\cfunc}(\cC P\cC).$
\end{lem}

\begin{proof}
Let $\{e_{1},...,e_{n}\}$ be an orthonormal basis of $\ran P$. Then
we have
\begin{align*}
\scC\left(\Q_{\cfunc}(P)\right) & =\scC\left(\left(\sum_{i=1}^{n}\Psi^{*}(e_{i})\Psi(e_{i})\right)-\frac{n}{2}\bbOne\right)=\sum_{i=1}^{n}\left(\Psi(\cC e_{i})\Psi^{*}(\cC e_{i})-\frac{1}{2}\bbOne\right)\\
 & =\sum_{i=1}^{n}\left(-\Psi^{*}(\cC e_{i})\Psi(\cC e_{i})+\frac{1}{2}\bbOne\right)=-\Q_{\cfunc}(\cC P\cC).
\end{align*}
\end{proof}
By Lemma \ref{lem:C(Q)}, seemingly it is natural to assume that every
$P\in\Scope(\cV)$ is $\cC$-invariant, i.e. $\cC P\cC=P$, in other
words, $\cC\ran P=\ran P$, because in this case we have $\scC(\Q_{\zeta}(P))=-\Q_{\zeta}(P)$,
so that the notion of ``charge'' appears to best fit with that of
``charge conjugation''. However, this assumption is not consistent
with the time evolution. For example, consider the Dirac field. Then
it is natural to think that the time evolution $P(t)$ ($t\in\R$)
of $P$ is given by $P(t):=e^{-\im tH/\hbar}Pe^{\im tH/\hbar}$ where
$H$ is the Dirac Hamiltonian defined by (\ref{eq:2D-Hamiltonian})
or (\ref{eq:4D-Hamiltonian}). However we see $\cC H\cC=-H$, and
hence $\cC P(t)\cC=(\cC P\cC)(-t)$ where $(\cC P\cC)(t):=e^{-\im tH/\hbar}\cC P\cC e^{\im tH/\hbar}$.%
{} Thus $\cC P\cC=P$ implies $\cC P(t)\cC=P(-t)\neq P(t)$ in general.
A simple alternative way is to understand that the usual charge conjugation
$\cC$ is not the same as the ``charge reversal'' $\tilde{\cC}$,
given by $P(t)\mapsto(\cC P\cC)(t)$, that is, the composition of
the charge conjugation and the time reversal $\cT:P(t)\mapsto P(-t)$.
However, more precisely we should write $\cT(P)(t):=P(-t)$, and so
$\cT$ is not defined for a single operator $P\in\FRP(\cV)$, but
for an operator-valued function $P(\cdot):\R\to\FRP(\cV)$; Especially,
there is no single operator $\cT$ on $\cV$ such that $P(-t)=\cT P(t)\cT^{-1}$
for all $t\in\R$, and the same applies to $\tilde{\cC}$. Probably,
it is not a natural manner to try to formulate the charge conjugation
(or the charge reversal) on the CAR algebra of a time-zero field.
Although it is evident that the spacetime-algebraic approach of Haag\textendash Araki\textendash Kastler
(i.e., of causal nets \cite{Haa96,Araki99}) is more natural for this
purpose, the general consideration of electric charge in terms of
causal nets seems extremely difficult.

For the (free) Dirac field, a compromise between the time-zero field
formalism and the spacetime field formalism can be given as follows.
For each $t\in\R$, let $\cV_{t}$ be a copy of $\cV$. More precisely
we can set $\cV_{t}:=\{F:\R\to\cV|t'\neq t\then F(t')=0\}$. For $t_{i}\in\R$
and $f_{i}\in\cV_{t_{i}}$, $i=1,2$, let $\id_{t_{1},t_{2}}:\cV_{t_{1}}\to\cV_{t_{2}}$
denote the canonical bijection. Then the time reversal $\cV_{t}\to\cV_{-t}$
is given by $\id_{t,-t}$. If each element of $\cV_{t}$ is expressed
as a function $F:\R\to\cV$, the time reversal $\cT$ is defined by
$\cT(F)(t):=F(-t)$. A ``scope at time $t$'' is identified as an
element $P$ of $\FRP(\cV_{t})$, and similarly $P$ can be seen as
a function $\R\to\FRP(\cV)$ such that $t'\neq t\then P(t')=0$. Thus
the time reversal $\cT(P)(t):=P(-t)$ makes sense, and hence we also
can define the ``charge reversal'' $\tilde{\cC}:P(t)\mapsto(\cC P\cC)(t)$. 

Consider the unital $C^{*}$-algebra $\widetilde{\CAR}(\cV)$ ``freely''
generated by the elements $\{\Psi(f)|f\in\cV_{t},t\in\R\}$ satisfying
the following anticommutation relations:
\[
\{\Psi(f_{1}),\Psi(f_{2})\}=\{\Psi^{*}(f_{1}),\Psi^{*}(f_{2})\}=0,\qquad\{\Psi(f_{1}),\Psi^{*}(f_{2})\}=\langle f_{1}|e^{\im(t_{2}-t_{1})H/\hbar}f_{2}\rangle\bbOne.
\]
Then the Dirac field is described in $\widetilde{\CAR}(\cV)$; In
fact, in the representation-independent algebraic approach to QFT
(e.g., \cite{BFV2003,BGP2007,BF2009,HW2015,KM2015,Ger2019}), such
an abstract algebra can be seen as the very \emph{definition} of the
Dirac field. Therefore, if $\CAR(\cV)$ is understood to be the ``time-zero
subalgebra'' of $\widetilde{\CAR}(\cV)$, one may expect that the
following can be justified:
\begin{lyxalgorithm}
\label{wh:C-inv}Let $P\in\FRP(\cV)$. Then $P\in\Scope(\cV)$ only
when $P$ is $\cC$-invariant (or equivalently, $\tilde{\cC}$-invariant).
\end{lyxalgorithm}

Then our definition of the ``scope-local charge'' $\Q_{\zeta}(\cdot)$
will be consistent with the conventional notion of charge conjugation
(or reversal). 

Note that there are $\cC$-invariant projections $P\in\FRP(\cV)$
such that $\Tr P$ is odd (e.g., $\Tr P=1$), and hence Working Hypothesis
\ref{wh:C-inv} does not explain our seemingly \emph{ad hoc} requirement
that $\Tr P$ is even (Working Hypothesis \ref{wh:even}). Next I
will make an attempt to explain that requirement.

\begin{defn}
\label{def:decomp}A projection $P\in\FRP(\cV)$ is called \termi{$\cC$-compatible}
if there exist $P_{\pm}\in\FRP(\cV)$ such that $P=P_{+}+P_{-}$ and
$\cC P_{\pm}\cC=P_{\mp}$.
\end{defn}

\begin{lyxalgorithm}
\label{wh:C-compat}Let $P\in\FRP(\cV)$. Then $P\in\Scope(\cV)$
only when $P$ is $\cC$-compatible.
\end{lyxalgorithm}

An intuitive and rough argument for this hypothesis is given as follows:
When we consider the Dirac field, if one wish to define the ``positron
part'' and the ``electron part'' of the ``scope'' $P$, it is
natural to define them to be $E_{+}PE_{+}$ and $E_{-}PE_{-}$, respectively
($E_{\pm}$ was defined by (\ref{eq:def:E+-})). However, $\Q_{\cfunc}(E_{\pm}PE_{\pm})$
is not a \emph{local} observable, while $\Q_{\cfunc}(P)$ is a local
observable for all $P\in\FRP(\cV)$. On the other hand, usually we
believe that the charge of a positron/electron is detectable \emph{locally},
e.g., as a trajectory in a cloud chamber. Hence one would like to
have an alternative definition of the \emph{local} positron/electron
parts of $P$. However notice that it is likely that the concept of
(spatially localizable) particle is no longer maintainable in relativistic
QFT, as I mentioned in Subsec.\,\ref{subsec:particle}.

This somewhat paradoxical situation have been wanting some explanation
with a logically and/or mathematically consistent formulation. One
of possible interpretations will be as follows: For each (local) positron/electron
detector, we have a \emph{temporary} (or \emph{ad hoc}) interpretation
that the detection can be decomposed \emph{locally} into the positron
and electron parts, so that the total detected charge is the sum of
the total charge of detected positrons and that of detected electrons.
Conversely, perhaps a (local) measurement of a quantity can be naturally
interpreted as a charge measurement only if it can be understood to
be decomposed locally into the positive-charge and negative-charge
parts. If this temporary interpretation is expressed as $P=P_{+}+P_{-}$,
the local observable $\d\Gamma(P_{+})$, whose spectrum is $\{0,1,...,n\}$
($n:=\Tr P_{+}=\Tr P_{-}=\frac{1}{2}\Tr P$), is interpreted as the
number (and also the total charge) of positrons in the scope $P$.
On the other hand, the local observable $n\bbOne-\d\Gamma(P_{-})$,
whose spectrum is $\{0,1,...,n\}$, is interpreted as the number of
the electrons in the scope $P$; Equivalently, $\d\Gamma(P_{-})-n\bbOne$
is interpreted as the total charge of electrons in $P$, so that the
total charge in $P$ is
\[
\Q_{\cfunc}(P)=\d\Gamma(P_{+})+\left(\d\Gamma(P_{-})-n\bbOne\right)=\d\Gamma(P)-n\bbOne.
\]
Thus we have $\cfunc(P)=\frac{1}{2}\Tr P$ for $P\in\Scope(\cV)$,
and so we can set $\cfunc(A):=\frac{1}{2}\Tr A$ for all $A\in\FR(\cV)$.\footnote{However note that for $A\notin\Scope(\cV)$, it is not clear whether
the value of $\cfunc(A)$ has some empirical meaning; While $\Q_{\cfunc}(A)$
is a local observable for all $A\in\FR(\cV)_{{\rm sa}}$, the meaning
of $\Q_{\cfunc}(A)$ is not immediately clear if $A\notin\Scope(\cV)$.
In any way, for the purposes of this paper, it is sufficient to determine
the value of $\cfunc(A)$ only for $A\in\Scope(\cV)$.} Although it seems that there exists no unique canonical decomposition
$P=P_{+}+P_{-}$, the definition of $\Q_{\cfunc}$ depends only on
$\cfunc$, not on the decomposition. In other words, Working Hypothesis
\ref{wh:C-compat} uniquely determines ``the total charge in the
scope $P$'', $\Q_{\cfunc}(P)$, independently of the decomposition
$P=P_{+}+P_{-}$.

For a fixed inertial observer, a relatively natural method of such
decomposition may be the eigenprojections $E_{\beta\pm}$ of $\beta=\gamma^{0}$,
seen as a multiplication operator on $\ol{\cV}=L^{2}(\R^{d},\C^{4}).$
We can understand $E_{\beta\pm}$ as the classical approximations
($\hbar\to0$), or the non-relativistic approximation $(c\to+\infty)$
of the projections $E_{\pm}$ defined by (\ref{eq:def:E+-}), viewed
from an inertial observer. If one dislikes to change the value the
universal constants $\hbar$ and $c$, instead one can consider the
low-momentum ($\|\p\|\approx0$) and/or the high-mass ($m\gg0$) approximations.
In this case, the scope $P$ should commute with $E_{\beta\pm}$,
and we can set $P_{\pm}=PE_{\beta\pm}$. However I will not argue
that this is the only possible way of the decomposition. Since $E_{\beta\pm}$
is not Lorentz invariant, we should not require the condition $[P,E_{\beta\pm}]=0$
as a general principle; This condition will be significant only for
a fixed inertial observer.

\subsection{Probabilistic law of scope-local charge}

Let $\cfunc(A):=\frac{1}{2}\Tr A$, and define $\Scope(\cV)$ to be
the set of $\cC$-compatible elements of $\FRP(\cV)$.

How can we describe the physical/empirical laws in terms of the scope-local
charges $\Q_{\zeta}(P)$ ($P\in\Scope(\cV)$)? I propose that it should
be described in terms of \emph{prior conditional probability} defined
below. This idea was developed in \cite{Yam2025,Yam2026b,Yam2026c},
especially for QFT in curved spacetime. An example of the physical/empirical
law stated with the notion of prior conditional probability shall
be given in Theorem \ref{thm:P(Q1|Q1)} below.

{} 

Recall that the \termi{algebra of observables in the scope $P$}
is the (finite-dimensional) $C^{*}$-subalgebra $\SYM{\fA_{P}}{AP}$
generated by $\bbOne$ and $\{\Psi^{*}(f)\Psi(g)|f,g\in\ran P\}$.
Since $\fA_{P}$ is isomorphic to a finite-dimensional matrix algebra,
the usual trace $\Tr$ of matrices induces the canonical trace $\Tr_{P}$
on $\fA_{P}$.

Let $P\in\Scope(\cV)$. Let $E_{1},...,E_{n}$ be projections in $\fA_{P}$.
Define the \termi{prior probability} $\Prob_{P}(E_{1}\cdots E_{n})$
in the scope $P$ by
\[
\Prob_{P}(A):=\frac{\Tr_{P}A^{*}A}{\Tr_{P}\bbOne},\qquad A:=E_{1}\cdots E_{n}.
\]
When we consider the limit $\Tr P\to\infty$ in some sense, the above
prior probability often converges to $0$, and so the value of $\Prob_{P}(E_{1},...,E_{n})$
is considered to have little empirical meaning in general. Empirically
more meaningful probability will be the following. Let $E_{1},...,E_{n}$
be projections in $\fA_{P}$, and $1\le k\le n$. Define the \termi{prior conditional probability}
by
\begin{equation}
\Prob_{P}(E_{k+1}\cdots E_{n}|E_{1}\cdots E_{k}):=\frac{\Prob_{P}(E_{1}\cdots E_{n})}{\Prob_{P}(E_{1}\cdots E_{k})},\label{eq:def:priorCondProb}
\end{equation}
that is,
\[
\Prob_{P}(B|A):=\frac{\Tr_{P}(AB)^{*}AB}{\Tr_{P}A^{*}A},\qquad A:=E_{1}\cdots E_{k},\ B:=E_{k+1}\cdots E_{n},
\]
when $\Tr_{P}A^{*}A>0$ (i.e., $A\neq0$).

For each $q\in\Spec(\Q_{\cfunc}(P))\subset\Z$, let $\SYM{\cE[\Q_{\cfunc}(P)=q]}{E[]}$
denote the spectral projection of $\Q_{\cfunc}(P)$ w.r.t.~$q$.
(Since $\fA_{P}$ is finite-dimensional, this spectral projection
always exists in $\fA_{P}$.) 

Let $P,P_{1},P_{2}\in\Scope(\cV)$, $P_{1},P_{2}\le P$, and $q_{i}\in\Spec(\Q_{\cfunc}(P_{i}))$.
Then we can consider the prior conditional probability
\[
\Prob_{P}(\Q_{\cfunc}(P_{2})=q_{2}|\Q_{\cfunc}(P_{1})=q_{1}):=\Prob_{P}(\E[\Q_{\cfunc}(P_{2})=q_{2}]|\E[\Q_{\cfunc}(P_{1})=q_{1}]).
\]
More generally, for $P,P_{i}\in\Scope(\cV)$, $P_{i}\le P$ ($i=1,...,n$),
and $q_{i}\in\Spec(\Q_{\cfunc}(P_{i}))$, we can define similarly
\[
\Prob_{P}(\Q_{\cfunc}(P_{k+1})=q_{k+1},...,\Q_{\cfunc}(P_{n})=q_{n}|\Q_{\cfunc}(P_{1})=q_{1},...,\Q_{\cfunc}(P_{k})=q_{k}).
\]

\begin{defn}
$P\in\Scope(\cV)$ is called an \termi{atomic scope} if $\Tr P=2$.
\end{defn}

The following lemmas are evident.
\begin{lem}
Let $P\in\Scope(\cV)$. Then the following conditions are equivalent.
\begin{enumerate}
\item $P$ is an atomic scope.
\item $P$ is minimal in $\Scope(\cV)\setminus\{0\}$.
\item There exists $P_{+}\in\FRP(\cV)$ such that $\Tr P_{+}=1$ and $P=P_{+}+P_{-}$,
$P_{-}:=\cC P_{+}\cC$.
\end{enumerate}
\end{lem}

\begin{lem}
Let $P\in\Scope(\cV)$ and $n:=\frac{1}{2}\Tr P$. Then there exist
atomic scopes $P_{1},...,P_{n}$ such that $P=P_{1}+\cdots+P_{n}$.
\end{lem}

\begin{lem}
If $P$ is an atomic scope, the spectrum of $\Q_{\zeta}(P)$ in $\fA:=\CAR(\cV)$
is given by $\Spec_{\fA}(\Q_{\zeta}(P))=\{-1,0,1\}$.
\end{lem}

Let $P,P_{1},P_{2}\in\Scope(\cV)$. Assume that $P_{i}$ is atomic,
and $P_{i}\le P$ for $i=1,2$. 

Let $P_{i}=P_{i+}+P_{i-}$, and $P_{i-}=\cC P_{i+}\cC$. Let $\SYM{f_{i\pm}}{fi+-}\in\ran(P_{\pm}),\ \|f\|=1,\ i=1,2.$
Then we have
\[
\d\Gamma(P_{i\pm})=\Psi^{*}(f_{i\pm})\Psi(f_{i\pm})=a_{i\pm}^{*}a_{i\pm},\qquad\SYM{a_{i\pm}}{ai+-}:=\Psi(f_{i\pm}).
\]
For $a_{i\pm},a_{i\pm}^{*}$, we see that the four operators $a_{i\pm}$
($i=1,2$) (and also the four operators $a_{i\pm}^{*}$ ($i=1,2$))
anticommute each other, and 
\begin{equation}
\{a_{i+}^{*},a_{i+}\}=\{a_{i-}^{*},a_{i-}\}=\bbOne,\qquad i=1,2.\label{eq:anticomm-ai+}
\end{equation}
Further, we have
\[
\Q_{\cfunc}(P_{i})=\d\Gamma(P_{i})-1=a_{i+}^{*}a_{i+}+a_{i-}^{*}a_{i-}-1,
\]
Consider the prior conditional probability $\Prob_{P}(\Q_{\cfunc}(P_{2})=1|\Q_{\cfunc}(P_{1})=1).$
We see
\[
\E[\Q_{\cfunc}(P_{i})=1]=a_{i+}^{*}a_{i+}a_{i-}^{*}a_{i-},
\]
\[
\E[\Q_{\cfunc}(P_{i})=-1]=a_{i+}a_{i+}^{*}a_{i-}a_{i-}^{*},
\]
\[
\E[\Q_{\cfunc}(P_{i})=0]=\bbOne-a_{i+}^{*}a_{i+}a_{i-}^{*}a_{i-}-a_{i+}a_{i+}^{*}a_{i-}a_{i-}^{*}.
\]
Hence we have
\[
\Prob_{P}(\Q_{\cfunc}(P_{2})=1|\Q_{\cfunc}(P_{1})=1)=\frac{\Tr_{P}(AB)^{*}AB}{\Tr_{P}A^{*}A},\qquad A:=\E[\Q_{\cfunc}(P_{1})=1],\ B:=\E[\Q_{\cfunc}(P_{2})=1],
\]
with
\[
\Tr_{P}A^{*}A=\Tr_{P}\E[\Q_{\cfunc}(P_{1})=1]=\Tr_{P}a_{1+}^{*}a_{1+}a_{1-}^{*}a_{1-}=2^{-2}\Tr_{P}\bbOne.
\]
\[
\Tr_{P}(AB)^{*}AB=\Tr_{P}AB=\Tr_{P}a_{1+}^{*}a_{1+}a_{1-}^{*}a_{1-}a_{2+}^{*}a_{2+}a_{2-}^{*}a_{2-}.
\]

The most trivial case is where $P_{1}=P_{2}$, so that $\Prob_{P}(\Q_{\cfunc}(P_{2})=q_{1}|\Q_{\cfunc}(P_{1})=q_{2})=\delta_{q_{1}q_{2}}$.

The second trivial case is where $P_{1}P_{2}=0$. In this case, we
have
\begin{equation}
\{a_{1+}^{*},a_{2+}\}=\{a_{1-}^{*},a_{2-}\}=\{a_{1+}^{*},a_{2-}\}=\{a_{1-}^{*},a_{2+}\}=0.\label{eq:anticomm-ai+-PP}
\end{equation}
We can check that
\[
\Tr_{P}a_{1+}^{*}a_{1+}a_{1-}^{*}a_{1-}a_{2+}^{*}a_{2+}a_{2-}^{*}a_{2-}=2^{-4}\Tr_{P}\bbOne,
\]
and hence
\[
\Prob_{P}(\Q_{\cfunc}(P_{2})=1|\Q_{\cfunc}(P_{1})=1)=\Prob_{P}(\Q_{\cfunc}(P_{2})=1)=\frac{1}{4}.
\]
Similarly we have
\[
\Prob_{P}(\Q_{\cfunc}(P_{2})=-1|\Q_{\cfunc}(P_{1})=-1)=\Prob_{P}(\Q_{\cfunc}(P_{2})=-1)=\frac{1}{4},\qquad\text{etc}.
\]
Roughly speaking, this implies that the two observables $\Q_{\cfunc}(P_{1})$
and $\Q_{\cfunc}(P_{2})$ are statistically independent if $P_{1}P_{2}=0$.%

More generally we have the following
\begin{thm}
\label{thm:P(Q1|Q1)}Let $P,P_{1},P_{2}\in\Scope(\cV)$, $P_{1},P_{2}$
atomic, and $P_{1},P_{2}\le P$. %
Then
\begin{equation}
\Prob_{P}(\Q_{\cfunc}(P_{2})=1|\Q_{\cfunc}(P_{1})=1)=\frac{1}{4}\det\left[\left(\bOne+P_{1}P_{2}\right)\upha\ran P_{1}\right].\label{eq:P(Q1|Q1)}
\end{equation}
\end{thm}

\begin{rem}
The r.h.s.~of (\ref{eq:P(Q1|Q1)}) depends only on the scopes $P_{1}$
and $P_{2}$, not on $P$. Thus it will be justified to replace ``$\Prob_{P}$''
in the l.h.s.~with ``$\Prob$'' simply. This $P$-independence
is a desirable property for a probabilistic law such as (\ref{eq:P(Q1|Q1)})
to be called a ``universal empirical law''; If a probabilistic law
is $P$-dependent, it should be understood to be ``context-dependent''
in a sense, and so non-universal.
\end{rem}

\begin{proof}
Let $P_{k}\in\FRP(\cV)$ ($k=1,2$) and $\Tr P_{i}=2$. Then the two
operators $P_{1}P_{2}P_{1}$ and $P_{2}P_{1}P_{2}$ have the same
eigenvalues $\{0,\lambda_{1},\lambda_{2}\}$ where $0\le\lambda_{1}\le\lambda_{2}\le1$.%
{} Note that
\[
\det\left[\left(\bOne+P_{1}P_{2}\right)\upha\ran P_{1}\right]=\left(1+\lambda_{1}\right)\left(1+\lambda_{2}\right).
\]
We will prove (\ref{eq:P(Q1|Q1)}) only for ``generic'' cases where
$0<\lambda_{1}\le\lambda_{2}<1$.

Let $\{e_{1,1},e_{1,2}\}$ be an orthonormal basis of $\ran P_{1}$
such that $P_{1}P_{2}e_{1,i}\,(=P_{1}P_{2}P_{1}e_{1,i})=\lambda_{i}e_{1,i}$,
$i=1,2$. Let
\[
e_{2,i}:=\frac{P_{2}e_{1,i}}{\|P_{2}e_{1,i}\|}=\lambda_{i}^{-1/2}P_{2}e_{1,i},\qquad i=1,2.
\]
Then we see that $\{e_{2,1},e_{2,2}\}$ is an orthonormal basis of
$\ran P_{2}$ such that $P_{2}P_{1}e_{2,i}=\lambda_{i}e_{2,i}$, $i=1,2$.
We also find that there exists an orthonormal system $f_{1},f_{2}\in\ker P_{1}$
such that
\[
e_{2,i}=\alpha_{i}e_{1,i}+\beta_{i}f_{i},\qquad\alpha_{i}:=\sqrt{\lambda_{i}},\ \beta_{i}:=\sqrt{1-\lambda_{i}}.
\]
To simplify the notations, let
\[
a_{i}:=\Psi(e_{1,i}),\qquad b_{i}:=\Psi(e_{2,i}),\qquad c_{i}:=\Psi(f_{i}),\qquad i=1,2.
\]
Note that
\begin{equation}
b_{i}=\Psi(\alpha_{i}e_{1,i}+\beta_{i}f_{i})=\alpha_{i}a_{i}+\beta_{i}c_{i}.\label{eq:bi=00003Dala}
\end{equation}
We have $\d\Gamma(P_{1})=\sum_{i=1}^{2}a_{i}^{*}a_{i}$, $\d\Gamma(P_{2})=\sum_{i=1}^{2}b_{i}^{*}b_{i}$,
and
\[
A:=\E[\Q_{\cfunc}(P_{1})=1]=a_{1}^{*}a_{1}a_{2}^{*}a_{2},\qquad B:=\E[\Q_{\cfunc}(P_{2})=1]=b_{1}^{*}b_{1}b_{2}^{*}b_{2}.
\]
By (\ref{eq:bi=00003Dala}) and $\alpha_{i}^{2}+\beta_{i}^{2}=1$,
we can verify that
\begin{align*}
\Tr_{P}(AB)^{*}AB & =\Tr_{P}AB=\Tr_{P}\left(a_{1}^{*}a_{1}a_{2}^{*}a_{2}b_{1}^{*}b_{1}b_{2}^{*}b_{2}\right)\\
 & =\alpha_{1}^{2}\alpha_{2}^{2}\Tr_{P}\left(a_{1}^{*}a_{1}a_{2}^{*}a_{2}a_{1}^{*}a_{1}a_{2}^{*}a_{2}\right)+\beta_{1}^{2}\alpha_{2}^{2}\Tr_{P}\left(a_{1}^{*}a_{1}a_{2}^{*}a_{2}c_{1}^{*}c_{1}a_{2}^{*}a_{2}\right)\\
 & \qquad+\alpha_{1}^{2}\beta_{2}^{2}\Tr_{P}\left(a_{1}^{*}a_{1}a_{2}^{*}a_{2}a_{1}^{*}a_{1}c_{2}^{*}c_{2}\right)+\beta_{1}^{2}\beta_{2}^{2}\Tr_{P}\left(a_{1}^{*}a_{1}a_{2}^{*}a_{2}c_{1}^{*}c_{1}c_{2}^{*}c_{2}\right)\\
 & =\alpha_{1}^{2}\alpha_{2}^{2}\Tr_{P}\left(a_{1}^{*}a_{1}a_{2}^{*}a_{2}\right)+\beta_{1}^{2}\alpha_{2}^{2}\Tr_{P}\left(a_{1}^{*}a_{1}a_{2}^{*}a_{2}c_{1}^{*}c_{1}\right)\\
 & \qquad+\alpha_{1}^{2}\beta_{2}^{2}\Tr_{P}\left(a_{1}^{*}a_{1}a_{2}^{*}a_{2}c_{2}^{*}c_{2}\right)+\beta_{1}^{2}\beta_{2}^{2}\Tr_{P}\left(a_{1}^{*}a_{1}a_{2}^{*}a_{2}c_{1}^{*}c_{1}c_{2}^{*}c_{2}\right)\\
 & =\frac{1}{4}\alpha_{1}^{2}\alpha_{2}^{2}\Tr_{P}\bbOne+\frac{1}{8}\beta_{1}^{2}\alpha_{2}^{2}\Tr_{P}\bbOne+\frac{1}{8}\alpha_{1}^{2}\beta_{2}^{2}\Tr_{P}\bbOne+\frac{1}{16}\beta_{1}^{2}\beta_{2}^{2}\Tr_{P}\bbOne\\
 & =\frac{1}{16}\left(4\alpha_{1}^{2}\alpha_{2}^{2}+2\beta_{1}^{2}\alpha_{2}^{2}+2\alpha_{1}^{2}\beta_{2}^{2}+\beta_{1}^{2}\beta_{2}^{2}\right)\Tr_{P}\bbOne\\
 & =\frac{1}{16}\left(\alpha_{1}^{2}+1\right)\left(\alpha_{2}^{2}+1\right)\Tr_{P}\bbOne\\
 & =\frac{1}{16}\left(\lambda_{1}+1\right)\left(\lambda_{2}+1\right)\Tr_{P}\bbOne\\
 & =\frac{1}{16}\det\left(\left(\bOne+P_{1}P_{2}\right)\upha\ran P_{1}\right)\Tr_{P}\bbOne
\end{align*}
Therefore, we obtain (\ref{eq:P(Q1|Q1)}) from $\Tr_{P}A^{*}A=\frac{1}{4}\Tr_{P}\bbOne$.
\end{proof}
Perhaps, more general probabilistic laws for scope-local charges can
be derived from Theorem \ref{thm:P(Q1|Q1)}, in principle. However,
it seems that the calculations needed to derive such laws for non-atomic
scopes become much more complicated, unless we can devise better calculation
techniques.

\def\arxiv#1{\href{https://arxiv.org/abs/#1} {\texttt{arXiv:#1}} }%
\def\ydoi#1{DOI: \href{https://doi.org/#1} {#1}}%

\newcommand{\etalchar}[1]{$^{#1}$}

\end{document}